\documentclass[a4paper,UKenglish,cleveref, autoref, thm-restate]{lipics-v2021}

\usepackage{multirow}

\usepackage{subcaption}

\usepackage{algorithm}
\usepackage[noend]{algpseudocode}

\newcommand{\D}{\textsf{D}}

\newcommand{\B}{\textsf{B}}

\newcommand{\cL}{\textsf{L}}
\newcommand{\cF}{\textsf{F}}

\newcommand{\mBWT}{\ensuremath{\mathsf{BWT}}}
\newcommand{\mC}{\ensuremath{\mathsf{C}}}

\def\rank{\mbox{\rm {\sf rank}}}

\def\select{\mbox{\rm {\sf select}}}

\def\LF{\mbox{\rm {\sf LF}}}

\def\SA{\mbox{\rm {\sf SA}}}

\def\BWT{\mbox{\rm {\sf BWT}}}

\def\P{\mbox{\rm {\sf P}}}
\def\D{\mbox{\rm {\sf D}}}

\def\B{\mbox{\rm {\sf B}}}

\def\W{\mbox{\rm {\sf W}}}

\newcommand{\ours}{\texttt{PFP-FM}}

\usepackage{xcolor}

\nolinenumbers

\title{Acceleration of FM-index Queries Through Prefix-free Parsing}

\author{Aaron Hong\footnote{Corresponding author}}{Department of Computer and Information Science and Engineering,  
		Herbert Wertheim College of Engineering,
		University of Florida,	Gainesville, FL, USA}{hong.yuc@ufl.edu}{[orcid]}{NIH/NHGRI grant R01HG011392 to Ben Langmead,  NSF/BIO grant DBI-2029552 to Christina Boucher}

\author{Marco Oliva}{Department of Computer and Information Science and Engineering,  
		Herbert Wertheim College of Engineering,
		University of Florida,	Gainesville, FL, USA}{marco.oliva@ufl.edu}{0000-0003-0525-3114}{NIH/NHGRI grant R01HG011392 to Ben Langmead, NSF/BIO grant DBI-2029552 to Christina Boucher}

\author{Dominik K\"{o}ppl}{Institut f\"{u}r Informatik der Univerist\"{a}t M\"{u}nster, 48149 M\"{u}nster, Germany}{koeppl.dsc@tmd.ac.jp}{0000-0002-7922-0369}{JSPS KAKENHI Grant Number JP21K17701, JP22H03551, and JP23H04378}
  
\author{Hideo~Bannai}{M\&D Data Science Center, Tokyo Medical and Dental University, Tokyo, Japan}{hdbn.dsc@tmd.ac.jp}{0000-0002-6856-5185}{JSPS KAKENHI Grant Number JP20H04141} 

\author{Christina Boucher}{Department of Computer and Information Science and Engineering,  
		Herbert Wertheim College of Engineering,
		University of Florida,	Gainesville, FL, USA}{christinaboucher@ufl.edu}{0000-0001-9509-9725}{NIH/NHGRI grant R01HG011392 to Ben Langmead,  NSF/BIO grant DBI-2029552 to Christina Boucher}

\author{Travis Gagie}{Faculty of Computer Science, Dalhousie University, Halifax, NS, Canada}{travis.gagie@gmail.com}{0000-0003-3689-327X}{NIH/NHGRI grant R01HG011392 to Ben Langmead, NSERC grant RGPIN-07185-2020 to Travis Gagie, NSF/BIO grant DBI-2029552 to Christina Boucher}

\authorrunning{A.\ Hong et al.}

\Copyright{Aaron Hong, Marco Oliva, Dominik K\"{o}ppl, Hideo Bannai, Christiana Boucher and Travis Gagie}



\category{}

\relatedversion{}

\EventEditors{John Q. Open and Joan R. Access}
\EventNoEds{2}
\EventLongTitle{42nd Conference on Very Important Topics (CVIT 2016)}
\EventShortTitle{CVIT 2016}
\EventAcronym{CVIT}
\EventYear{2016}
\EventDate{December 24--27, 2016}
\EventLocation{Little Whinging, United Kingdom}
\EventLogo{}
\SeriesVolume{42}
\ArticleNo{23}

\begin{document}

\maketitle

\begin{abstract}
FM-indexes are a crucial data structure in DNA alignment, for example, but searching with them usually takes at least one random access per character in the query pattern.  Ferragina and Fischer~\cite{ferragina2007suffix} observed in 2007 that word-based indexes often use fewer random accesses than character-based indexes, and thus support faster searches.  Since DNA lacks natural word-boundaries, however, it is necessary to parse it somehow before applying word-based FM-indexing.  Last year, Deng et al.~\cite{deng2022fm} proposed parsing genomic data by induced suffix sorting, and showed the resulting word-based FM-indexes support faster counting queries than standard FM-indexes when patterns are a few thousand characters or longer.  In this paper we show that using prefix-free parsing---which takes parameters that let us tune the average length of the phrases---instead of induced suffix sorting, gives a significant speedup for patterns of only a few hundred characters.  We implement our method and demonstrate it is between 3 and 18 times faster than competing methods on queries to GRCh38.  And was consistently faster on queries made to 25,000, 50,000 and 100,000 SARS-CoV-2 genomes. Hence, it is very
clear that our method accelerates the performance of count over all state-of-the-art methods with a minor increase in the memory.
The source code for $\ours$ is available at \url{https://github.com/marco-oliva/afm}.
\end{abstract}  

\section{Introduction}
\label{sec:intro}

The FM-index \cite{ferragina_indexing_2005} is one of the most famous data structures in bioinformatics as it has been applied to countless applications in the analysis of biological data. Due to the long-term impact of this data structure, Burrows, Ferragina, and Manzini earned the 2022 ACM Paris Kanellakis Theory and Practice Award\footnote{https://awards.acm.org/kanellakis}.  It is the data structure behind important read aligners --- e.g., Bowtie \cite{langmead_ultrafast_2009} and BWA \cite{bwa} --- which take one or more reference genomes and build the FM-index for these genomes and use the resulting index to find short exact alignments between a set of reads and the reference(s) which then can be extended to approximate matches  \cite{langmead_ultrafast_2009,bwa}.  Briefly, the FM-index consists of a sample of the suffix array (denoted as \SA{}) and the Burrows--Wheeler transform (\BWT{}) array. Given an input string $S$ and a query pattern $Q$, {\tt count} queries that answer the number of times the longest match of $Q$ appears in $S$, can be efficiently supported using the \BWT{}.  To locate all of these occurrences the \SA{} sample is needed.  Hence, together the FM-index efficiently supports both {\tt count} and {\tt locate} queries.  We mathematically define the \SA{} and \BWT{} in the next section.

There has been a plethora of research papers on reducing the size of the FM-index (see, e.g.,~\cite{makinen2005succinct,gog2019fixed,gagie2020fully}) and on speeding up queries.  The basic query, {\tt count}, returns the number of times a pattern $Q$ appears in the indexed text $S$, but usually requires at least $|Q|$ random accesses to the \BWT{} of $S$, which are usually much slower than the subsequent computations we perform on the information those accesses return.  More specifically, a {\tt count} query for $Q$ use $\rank$ queries at $|Q|$ positions in the \BWT{}; if we answer these using a single wavelet tree for the whole \BWT{}, then we may use a random access for every level we descend in the wavelet tree, or $\Omega (|Q| \log \sigma)$ random access in all, where $\sigma$ is the size of the alphabet; if we break the \BWT{} into blocks and use a separate wavelet tree for each block~\cite{gog2019fixed}, we may need only one or a few random accesses per $\rank$ query, but the total number of random accesses is still likely to be $\Omega (|Q|)$.  As far back as 2007, Ferragina and Fischer~\cite{ferragina2007suffix} addressed compressed indexes' reliance on random access and demonstrated that word-based indexes perform fewer random accesses than character-based indexes:
{\it
``The space reduction of the final word-based suffix array impacts also in their query time (i.e. less random access binary-search steps!), being faster by a factor of up to 3.''
}

Thus, one possibility of accelerating the random access to genomic data---where it is widely used---is to break up the sequences into words or phrases.  In light of this insight, Deng et al. \cite{deng2022fm} in 2022 used the Longest Matching Suffix (LMS) factorization~\cite{daykin2013linear}  to break an input sequence $S$ into phrases. Unfortunately, one round of LMS parsing leads to phrases that are generally too short, so they obtained speedup only when $Q$ was thousands of characters.  The open problem was how to control the length of phrases with respect to the input to get longer phrases that would enable larger advances in the acceleration of the random access. 

Here, we apply the concept of prefix-free parsing to the problem of accelerating {\tt count} in the FM-index.  
Prefix-free parsing uses a rolling hash to first select trigger strings that are then used to define a parse of the input string $S$; i.e., the prefix-free parse is a parsing of $S$ into phrases that begin and end at trigger strings and contain no other trigger string.  All unique phrases are lexicographically sorted and stored in the dictionary of the prefix-free parse, which we denote as $\D$.  The prefix-free parse can be stored as an ordered list of the phrases' ranks in $\D$.   Hence, prefix-free parsing breaks up the input sequence into phrases, whose size is more controllable by the selection of the trigger strings. This leads to a more flexible acceleration than Deng et al.~\cite{deng2022fm} obtained.  

Now suppose we build an FM-index for $S$, an FM-index for the parse $\P$, and a bitvector $\B$ of length $n$ with 1s marking characters in the \BWT{} of $S$ that immediately precede phrase boundaries in $S$, i.e., that immediately precede a trigger string. We note that all the 1s are bunched into at most as many runs as there are distinct trigger strings in $S$.  Also, as long as the ranks of the phrases are in the same lexicographic order as the phrases themselves, we can use the bitvector to map from the interval in the \BWT{} of $S$ for any pattern starting with a trigger string to the corresponding interval in the \BWT{} of $\P$, and vice versa. This means that, given a query pattern $Q$, we can backward search for $Q$ character by character in the FM-index for $S$ until we hit the left end of the rightmost trigger string in $Q$, then map into the \BWT{} of $\P$ and backward search for $Q$ phrase by phrase until we hit the left end of the leftmost trigger string in $Q$, then map back into the \BWT{} of $S$ and finish backward searching character by characters again.

We implement this method, which we refer to as $\ours$, and extensively compare against the FM-index implementation in {\tt sdsl} \cite{gbmp2014sea}, {\tt RLCSA}~\cite{siren2009compressed}, {\tt RLFM}~\cite{makinen2005succinct,makinen2004run}, and {\tt FIGISS} \cite{deng2022fm} using sets of SARS-CoV-2 genomes taken from the NCBI website, and the Genome Reference Consortium Human Build 38 with varying query string lengths. When we compare $\ours$ to FM-index in {\tt sdsl} using 100,000 SARS-CoV-2 genomes, we witnessed that $\ours$ was able to perform between 2.1 and 2.8 more queries.   In addition, $\ours$ was between 64.38\% and 74.12\%, 59.22\% and 78.23\%, and 49.10\% and 90.70\% faster than {\tt FIGISS},  {\tt RLCSA}, and {\tt RLFM}, respectively on 100,000 SARS-CoV-2 genomes. We evaluated the performance of $\ours$ on the Genome Reference Consortium Human Build 38, and witnessed that it was between  3.86 and 7.07, 2.92 and 18.07, and 10.14 and 25.46 times faster than {\tt RLCSA}, {\tt RLFM}, and {\tt FIGISS}, respectively. With respect to construction time, $\ours$ had the most efficient construction time for all SARS-CoV-2 datasets and was the second fastest for Genome Reference Consortium Human Build 38. All methods used less than 60 GB for memory for construction on the SARS-CoV-2 datasets, making the construction feasible on any entry level commodity server---even the build for the 100,000 SARS-CoV-2 dataset. Construction for the Genome Reference Consortium Human Build 38 required between 26 GB and 71 GB for all methods, with our method using the most memory.  In summary, we develop and implement a method for accelerating the FM-index, and achieve an acceleration between 2 and 25 times, with the greatest acceleration witnessed with longer patterns.  Thus, accelerated FM-index methods---such as the one developed in this paper---are highly applicable to finding very long matches (125 to 1,000 in length) between query sequences and reference databases.  As reads get longer and more accurate (i.e., Nanopore data), we will soon be prepared align long reads to reference databases with efficiency that surpasses traditional FM-index based alignment methods.  The source code is publicly available at \url{https://github.com/marco-oliva/afm}.
\section{Preliminaries} \label{sec:prelimi}

\subsection{Basic Definitions} A string $S$ of length $n$ is a finite sequence of symbols $S = S[0..n-1] = S[0] \cdots S[n-1]$ over an alphabet $\Sigma = \{c_1, \ldots , c_{\sigma} \}$. We assume that the symbols can be unambiguously ordered.  We denote by $\varepsilon$ the empty string, and the length of $S$ as $|S|$.  Given a string $S$, we denote the reverse of $S$ as $rev(S)$, i.e., $rev(S) = S[n-1] \cdots S[0]$.

We denote by $S[i..j]$ the substring $S[i] \cdots  S[j]$ of $S$ starting in position $i$ and ending in position $j$, with $S[i .. j] = \varepsilon$ if $i > j$. For a string $S$ and $0 \leq i < n$, $S[0 .. i]$ is called the $i$-th prefix of $S$, and $S[i..n-1]$ is called the $i$-th suffix of $S$. We call a prefix $S[0 .. i]$ of $S$ a {\em proper prefix} if $0 \leq i < n-1$.  Similarly, we call a suffix $S[i .. n-1]$ of $S$ a {\em proper suffix} if $0 < i < n$.

Given a string $S$, a symbol $c \in \Sigma$, and an integer $i$, we define $S.\rank_c(i)$ (or simply {\tt rank} if the context is clear) as the number of occurrences of $c$ in $S[0..i-1]$. 
We also define $S.\select_c(i)$ as 
$\min(\{ j-1 \mid S.\rank_c(j) = i \}\cup\{ n \})$,
i.e., the position in $S$ of the $i$-th occurrence of $c$ in $S$ if it exists, and $n$ otherwise.  For a bitvector $B[0..n-1]$, that is a string over $\Sigma = \{0,1\}$, to ease the notation we will refer to $B.\rank_1(i)$ and $B.\select_1(i)$ as $B.\rank(i)$ and $B.\select(i)$, respectively.

\subsection{SA, BWT, and Backward Search}

We denote the {\em suffix array}~\cite{mm1993} of a given a string $S[0..n-1]$ as $\SA_S$, and define it to be the permutation of $\{0,\ldots,n-1\}$ such that $S[\SA_S[i]..n-1]$ is the $i$-th lexicographical smallest suffix of $S$.  We refer to $\SA_S$ as $\SA$  when it is clear from the context.
For technical reasons, we assume that the last symbol of the input string is $S[n-1] = \$$, which does not occur anywhere else in the string and is smaller than any other symbol.

We consider the matrix $\W$ containing all sorted rotations of $S$, called the \BWT matrix of $S$, and let \cF{} and \cL{} be the first and the last column of the matrix.  
The last column defines the \BWT{} array, i.e., $\BWT = \cL$.  
Now let $\mC[c]$ be the number of suffixes starting with a character smaller than $c$. We define the \LF{}-mapping as  $\LF(i,c) = \mC[c] + \rank_c(\mBWT,i)$  and  $\LF(i) = \LF(i,\mBWT[i])$. 
With the \LF{}-mapping, it is possible to reconstruct the string $S$ from its \BWT{}.  It is in fact sufficient to set an iterator $s = 0$ and $S[n-1] = \$$ and for each $i = n-2, \ldots, 0$ do $S[i] = \BWT[s]$ and $s = \LF(s)$. The \LF{}-mapping can also be used to support {\tt count} by performing the backward search, which we now describe.

Given a query pattern $Q$ of length $m$, the \emph{backward search} algorithm consists of $m$ steps that preserve the following invariant: at the $i$-th step, $p$ stores the position of the first row of $\W$ prefixed by $Q[i,m]$ while $q$ stores the position of the last row of $\W$ prefixed by $Q[i,m]$.  To advance from $i$ to $i-1$, we use the \LF{}-mapping on $p$ and $q$, $p = \mC[c] + \mBWT.\rank_c(p)$ and $q = \mC[c] + \BWT{}.\rank_c(q + 1) - 1$.

\subsection{FM-index and count Queries}

Given a query string $Q[0..m-1]$ and an input string $S[0..n-1]$, two fundamental queries are: (1) {\tt count} which counts the number of of occurrences of $Q$ in $S$; (2) {\tt locate} which finds the location of each of these matches in $S$.  Ferragina and Manzini \cite{ferragina_indexing_2005} showed that, by combining $\SA$ with the \BWT{}, both {\tt count} and {\tt locate} can be efficiently supported.  Briefly, backward search on the \BWT{} is used to find the lexicographical range of the occurrences of $Q$ in $S$; the size of this range is equal to {\tt count}.  The \SA{} positions within this range are the positions where these occurrences are in $S$.

\subsection{Prefix-Free Parsing} 

As we previously mentioned, the {\em Prefix-Free Parsing} (PFP) takes as input a string $S[0..n-1]$, and positive integers $w$ and $p$, and produces a parse of $S$ (denoted as $\P$) and a dictionary (denoted as $\D$) of all the unique substrings (or phrases) of the parse. We briefly go over the algorithm for producing this dictionary and parse.  First, we let $T$ be an arbitrary set of $w$-length strings over $\Sigma$
and call it the set of {\em trigger strings}. 
As mentioned before, we assume that $S[n-1]=\$$ and consider $S$ to be cyclic, i.e., for all $i$, $S[i] = S[i \bmod n]$. Furthermore, we assume that 
$\$S[0..w-2]=S[n-1..n+w-2] \in T$, i.e., the substring of length $w$ that begins with $\$$ is a trigger string.

We let the dictionary $\D = \{d_1,.., d_{\big|\D\big|}\}$ be a (lexicographically sorted) maximum set of substrings of $S$ such that the following holds for each $d_i$: i) exactly one proper prefix of $d_i$ is contained in $T$, ii) exactly one proper suffix of $d_i$ is contained in $T$, iii) and no other substring of $d_i$ is in $T$.   
An important consequence of the definition is that $\D$ is prefix-free, i.e., for any $i \neq j$, $d_i$ cannot be a prefix of $d_j$.

Since we assumed $S[n-1..n+w-2] \in T$, we can construct $\D$ by scanning $S' = \$S[0..n-2]S[n-1..n+w-2]$ to find all occurrences of $T$ and adding to $\D$ each substring of $S'$ that starts and ends at a trigger string being inclusive of the starting and ending trigger string.
We can also construct the list of occurrences of $\D$ in $S'$, which defines the parse $\P$.

We choose $T$ by a Karp-Rabin fingerprint $f$ of strings of length $w$.
We slide a window of length $w$ over $S'$, and for each length $w$ substring $r$ of $S'$,
include $r$ in $T$ if and only if $f(r) \equiv 0\pmod p$ or $r=S[n-1..n+w-2]$.
Let $0=s_0 < \cdots < s_{k-1}$ be the positions in $S'$ such that for any $0 \le i < k$, $S'[s_i ..s_i+w-1] \in T$.
The dictionary is $\D = \{ S'[s_i..s_{i+1}+w-1]\mid i=0,\ldots, k-1\}$,
and the parse is defined to be the sequence of lexicographic ranks in $\D$ of the substrings
$S'[s_0..s_{1}+w-1], \ldots, S'[s_{k-2}..s_{k-1}+w-1]$.

As an example, suppose we have, $S' = {\tt {\color{green}{\$A}}G{\color{red}AC}G{\color{red}AC}T\#AGAT{\color{red}AC}T\#AGAT{\color{blue}TC}GAG{\color{red}AC}G{\color{red}AC}{\color{green}\$A}}$, where the trigger strings are highlighted in red, blue, or green.  Then, assuming $\$ < \# < \tt{A} < \tt{C} < \tt{G} < \tt{T}$,
we have
$
\D= 
\{ \tt{\$AGAC}, \tt{AC\$A}, \tt{ACGAC}, \tt{ACT\#AGATAC}, \tt{ACT\#AGATTC}, \tt{TCGAGAC}\}
$
and
$\P= 0, 2, 3, 4, 5, 2, 1$.
\section{Methods}\label{sec:methods}

As we previously mentioned, we will use prefix-free parsing to build a word-based FM-index in a manner in which the length of the phrases can be controlled via the parameters $w$ and $p$.  To explain our data structure, we first describe the various components of our data structure, and then follow with describing how to support {\tt count} queries in a manner that is more efficient than the standard FM-index.

\subsection{Data Structure Design}

It is easiest to explain our two-level design with an example, so consider a text
\[S [0..n - 1] = \mathtt{TCCAGAAGAGTATCTCCTCGACATGTTGAAGACATATGAT\$}\]
of length $n = 41$ that is terminated by a special end-of-string character {\tt \$} lexicographically less than the rest of the alphabet.  Suppose we parse $S$ using $w = 2$ and a Karp-Rabin hash function such that the normal trigger strings occurring in $S$ are {\tt AA}, {\tt CG} and {\tt TA}.  We consider $S$ as cyclic, and we have  $\$ S[0..w - 2] = \mathtt{\$ T}$ as a special trigger string, so the the dictionary $\D$ is
\[\D [0..5] = \{\mathtt{\$TCCAGAA}, \mathtt{AAGACATA}, \mathtt{AAGAGTA}, \mathtt{CGACATGTTGAA}, \mathtt{TATCTCCTCG}, \mathtt{TATGAT\$T}\}\,,\]
with the phrases sorted in lexicographic order.  (Recall that phrases consecutive in $S$ overlap by $w = 2$ characters.)  If we start parsing at the {\tt \$}, then the prefix-free parse for $S$ is
\[\P [0..5] = (0, 2, 4, 3, 1, 5)\,,\]
where each element (or phrase ID) in $\P$ is the lexicographic rank of the phrase in $\D$.

Next, we consider the \BWT{} matrix for $\P$.  Figure~\ref{fig:P_matrix} illustrates the \BWT{} matrix of \P{} for our example. We note that since there is only one {\tt \$} in $S$, it follows that there is only one 0 in $\P$; 
we can regard this $0$ as the end-of-string character for (a suitable rotation of) $P$ corresponding to {\tt \$} in $S$.
If we take the $i$-th row of this matrix and replace the phrase IDs by the phrases themselves, collapsing overlaps, then we get the lexicographically $i$-th cyclic shift of $S$ that start with a trigger string, as shown on the right of the figure.  
This is one of the key insights that we will use later on.
\begin{lemma}\label{lemma:lexordercorrespondence}
The lexicographic order of rotations of $\P$ correspond to the lexicographic order of their corresponding rotations of $S$.
\end{lemma}
\begin{proof}
The characters of $\P$ are the phrase IDs that act as meta-characters.
Since the meta-characters inherit the lexicographic rank of their underlying characters, 
and due to the prefix-freeness of the phrases,
the suffix array of $\P$ permutes the meta-characters of $\P$ in the same way as the suffix array of $S$ the phrases of $S$.
    This means that the order of the phrases in the $\BWT$ of $S$ is the same as their corresponding phrase IDs in $\P$.
\end{proof}
Next, we let $B [0..n - 1]$ be a bitvector marking these cyclic shifts' lexicographic rank among all cyclic shifts of $S$, i.e., where they are among the rows of the \BWT{} matrix of $S$.  Figure~\ref{fig:T_matrix} shows the \SA{}, \BWT{} matrix and \BWT{} of $S$, together with $B$; we highlight the \BWT{}---the last column of the matrix---and the cyclic shifts from Figure~\ref{fig:P_matrix} in red.  We note that $B$ contains at most one run of 1's for each distinct trigger string in $S$ so it is usually highly run-length compressible in practice.

\begin{figure}[t!]
\begin{center}
\begin{tabular}{c@{\hspace{10ex}}c}
\begin{tabular}{rrrrrr}
0 & 2 & 4 & 3 & 1 & 5\\
1 & 5 & 0 & 2 & 4 & 3\\
2 & 4 & 3 & 1 & 5 & 0\\
3 & 1 & 5 & 0 & 2 & 4\\
4 & 3 & 1 & 5 & 0 & 2\\
5 & 0 & 2 & 4 & 3 & 1
\end{tabular}
&
\begin{tabular}{c}
\tt \$TCCAGAAGAGTATCTCCTCGACATGTTGAAGACATATGAT\\
\tt AAGACATATGAT\$TCCAGAAGAGTATCTCCTCGACATGTTG\\
\tt AAGAGTATCTCCTCGACATGTTGAAGACATATGAT\$TCCAG\\
\tt CGACATGTTGAAGACATATGAT\$TCCAGAAGAGTATCTCCT\\
\tt TATCTCCTCGACATGTTGAAGACATATGAT\$TCCAGAAGAG\\
\tt TATGAT\$TCCAGAAGAGTATCTCCTCGACATGTTGAAGACA
\end{tabular}
\end{tabular}
\caption{The \BWT{} matrix for our prefix-free parse $P$ {\bf (left)} and the cyclic shifts of $S$ that start with a trigger string {\bf (right)}, in lexicographic order.}
\label{fig:P_matrix}
\end{center}
\end{figure}

\begin{figure}[p!]
\begin{center}
\begin{tabular}{rrcc}
$i$ & $\SA [i]$ & $B [i]$ & $T [\SA [i]..(\SA [i] - 1) \bmod n]$ \hfill \textcolor{red}{$\BWT [i]$}\\
\hline
 0 & 40 & 1 & \tt \textcolor{red}{\$TCCAGAAGAGTATCTCCTCGACATGTTGAAGACATATGAT}\\
 1 & 28 & 1 & \tt \textcolor{red}{AAGACATATGAT\$TCCAGAAGAGTATCTCCTCGACATGTTG}\\
 2 &  5 & 1 & \tt \textcolor{red}{AAGAGTATCTCCTCGACATGTTGAAGACATATGAT\$TCCAG}\\
 3 & 31 & 0 & \tt ACATATGAT\$TCCAGAAGAGTATCTCCTCGACATGTTGAA\textcolor{red}{G}\\
 4 & 20 & 0 & \tt ACATGTTGAAGACATATGAT\$TCCAGAAGAGTATCTCCTC\textcolor{red}{G}\\
 5 &  3 & 0 & \tt AGAAGAGTATCTCCTCGACATGTTGAAGACATATGAT\$TC\textcolor{red}{C}\\
 6 & 29 & 0 & \tt AGACATATGAT\$TCCAGAAGAGTATCTCCTCGACATGTTG\textcolor{red}{A}\\
 7 &  6 & 0 & \tt AGAGTATCTCCTCGACATGTTGAAGACATATGAT\$TCCAG\textcolor{red}{A}\\
 8 &  8 & 0 & \tt AGTATCTCCTCGACATGTTGAAGACATATGAT\$TCCAGAA\textcolor{red}{G}\\
 9 & 38 & 0 & \tt AT\$TCCAGAAGAGTATCTCCTCGACATGTTGAAGACATAT\textcolor{red}{G}\\
10 & 33 & 0 & \tt ATATGAT\$TCCAGAAGAGTATCTCCTCGACATGTTGAAGA\textcolor{red}{C}\\
11 & 11 & 0 & \tt ATCTCCTCGACATGTTGAAGACATATGAT\$TCCAGAAGAG\textcolor{red}{T}\\
12 & 35 & 0 & \tt ATGAT\$TCCAGAAGAGTATCTCCTCGACATGTTGAAGACA\textcolor{red}{T}\\
13 & 22 & 0 & \tt ATGTTGAAGACATATGAT\$TCCAGAAGAGTATCTCCTCGA\textcolor{red}{C}\\
14 &  2 & 0 & \tt CAGAAGAGTATCTCCTCGACATGTTGAAGACATATGAT\$T\textcolor{red}{C}\\
15 & 32 & 0 & \tt CATATGAT\$TCCAGAAGAGTATCTCCTCGACATGTTGAAG\textcolor{red}{A}\\
16 & 21 & 0 & \tt CATGTTGAAGACATATGAT\$TCCAGAAGAGTATCTCCTCG\textcolor{red}{A}\\
17 &  1 & 0 & \tt CCAGAAGAGTATCTCCTCGACATGTTGAAGACATATGAT\$\textcolor{red}{T}\\
18 & 15 & 0 & \tt CCTCGACATGTTGAAGACATATGAT\$TCCAGAAGAGTATC\textcolor{red}{T}\\
19 & 18 & 1 & \tt \textcolor{red}{CGACATGTTGAAGACATATGAT\$TCCAGAAGAGTATCTCCT}\\
20 & 13 & 0 & \tt CTCCTCGACATGTTGAAGACATATGAT\$TCCAGAAGAGTA\textcolor{red}{T}\\
21 & 16 & 0 & \tt CTCGACATGTTGAAGACATATGAT\$TCCAGAAGAGTATCT\textcolor{red}{C}\\
22 & 27 & 0 & \tt GAAGACATATGAT\$TCCAGAAGAGTATCTCCTCGACATGT\textcolor{red}{T}\\
23 &  4 & 0 & \tt GAAGAGTATCTCCTCGACATGTTGAAGACATATGAT\$TCC\textcolor{red}{A}\\
24 & 30 & 0 & \tt GACATATGAT\$TCCAGAAGAGTATCTCCTCGACATGTTGA\textcolor{red}{A}\\
25 & 19 & 0 & \tt GACATGTTGAAGACATATGAT\$TCCAGAAGAGTATCTCCT\textcolor{red}{C}\\
26 &  7 & 0 & \tt GAGTATCTCCTCGACATGTTGAAGACATATGAT\$TCCAGA\textcolor{red}{A}\\
27 & 37 & 0 & \tt GAT\$TCCAGAAGAGTATCTCCTCGACATGTTGAAGACATA\textcolor{red}{T}\\
28 &  9 & 0 & \tt GTATCTCCTCGACATGTTGAAGACATATGAT\$TCCAGAAG\textcolor{red}{A}\\
29 & 24 & 0 & \tt GTTGAAGACATATGAT\$TCCAGAAGAGTATCTCCTCGACA\textcolor{red}{T}\\
30 & 39 & 0 & \tt T\$TCCAGAAGAGTATCTCCTCGACATGTTGAAGACATATG\textcolor{red}{A}\\
31 & 10 & 1 & \tt \textcolor{red}{TATCTCCTCGACATGTTGAAGACATATGAT\$TCCAGAAGAG}\\
32 & 34 & 1 & \tt \textcolor{red}{TATGAT\$TCCAGAAGAGTATCTCCTCGACATGTTGAAGACA}\\
33 &  0 & 0 & \tt TCCAGAAGAGTATCTCCTCGACATGTTGAAGACATATGAT\textcolor{red}{\$}\\
34 & 14 & 0 & \tt TCCTCGACATGTTGAAGACATATGAT\$TCCAGAAGAGTAT\textcolor{red}{C}\\
35 & 17 & 0 & \tt TCGACATGTTGAAGACATATGAT\$TCCAGAAGAGTATCTC\textcolor{red}{C}\\
36 & 12 & 0 & \tt TCTCCTCGACATGTTGAAGACATATGAT\$TCCAGAAGAGT\textcolor{red}{A}\\
37 & 26 & 0 & \tt TGAAGACATATGAT\$TCCAGAAGAGTATCTCCTCGACATG\textcolor{red}{T}\\
38 & 36 & 0 & \tt TGAT\$TCCAGAAGAGTATCTCCTCGACATGTTGAAGACAT\textcolor{red}{A}\\
39 & 23 & 0 & \tt TGTTGAAGACATATGAT\$TCCAGAAGAGTATCTCCTCGAC\textcolor{red}{A}\\
40 & 25 & 0 & \tt TTGAAGACATATGAT\$TCCAGAAGAGTATCTCCTCGACAT\textcolor{red}{G}
\end{tabular}
\caption{The SA, BWT matrix and BWT of $T$, together with the bitvector $B$ in which 1s indicate rows of the matrix starting with trigger strings.  The BWT is highlighted in red, as are the columns marked by 1s.}
\label{fig:T_matrix}
\end{center}
\end{figure}

In addition to the bitvector, we store a hash function $h$ on phrases and a map $M$ from the hashes of the phrases in $\D$ to those phrases' lexicographic ranks, which are their phrase IDs; $M$ returns NULL when given any other key.  Therefore, in total, we build the FM-index for $S$, the FM-index for $\P$, the bitvector $B$ marking the cyclic rotations, the hash function $h$ on the phrases and the map $M$. For our example, suppose

\[\begin{array}{c@{\hspace{10ex}}c}
\begin{array}{rcl}
h (\mathtt{\$TCCAGAA}) & = & 91785\\
h (\mathtt{AAGACATA}) & = & 34865\\
h (\mathtt{AAGAGTA}) & = & 49428\\
h (\mathtt{CGACATGTTGAA}) & = & 98759\\
h (\mathtt{TATCTCCTCG}) & = & 37298\\
h (\mathtt{TATGAT\$T}) & = & 68764
\end{array}
&
\begin{array}{rcl}
M (91785) & = & 0\\
M (34865) & = & 1\\
M (49428) & = & 2\\
M (98759) & = & 3\\
M (37298) & = & 4\\
M (68764) & = & 5
\end{array}
\end{array}\]
and $M (x) = \mathrm{NULL}$ for any other value of $x$.  

If we choose the range of $h$ to be reasonably large then we can still store $M$ in space proportional to the number of phrases in $D$ with a reasonably constant coefficient and evaluate $M (h (\cdot))$ in constant time with high probability, but the probability is negligible that $M (h (\gamma)) \neq \mathrm{NULL}$ for any particular string $\gamma$ not in $D$.  This means that in practice we can use $M (h (\cdot))$ as a membership dictionary for $D$, and not store $D$ itself.

\subsection{Query Support}

Next, given the data structure that we define above, we describe how to support {\tt count} queries for a given pattern $Q$. We begin by parsing $Q$ using the same Karp-Rabin hash we used to parse $S$, implying that we will have all the same trigger strings as we did before but possibly additional ones that did not occur in $S$.
However, we will not consider $Q$ to be cyclic nor assume an end-of-string symbol that would assure that $Q$ starts and ends with a trigger string.

If $Q$ is a substring of $S$, then, since $Q$ contains the same trigger strings as its corresponding occurrence in $S$,
the sequence of phrases induced by the trigger strings in
$Q$ must be a substring of the sequence of phrases of $S$.
Together with the prefix and suffix of $Q$ that are a suffix and prefix of the phrases in $S$ to the left and right of the shared phrases, we call this the partial encoding of $Q$, defined formally as follows.
\begin{definition}[partial encoding]
Given a substring $S[i..j]$ of $S$, the
{\em partial encoding} of $S[i..j]$ is defined as follows:
If no trigger string occurs in $S[i..j]$,
then the partial encoding of $S[i..j]$ is simply $S[i..j]$ itself.
Otherwise, the partial encoding of $S[i..j]$ is the concatenation of:
(1) the shortest prefix $\alpha$ of $S[i..j]$ that does not start with a trigger string and ends with a trigger string,
followed by 
(2) the sequence of phrase IDs of phrases completely contained in $S[i..j]$,
followed by 
(3) the shortest suffix $\beta$ of $S[i..j]$ that begins with a trigger string and does not end with a trigger string.
\end{definition}
So the partial encoding partitions $S[i..j]$ into a prefix $\alpha$, a list of phrase IDs, and a suffix $\beta$.
If $S[i..j]$ begins (respectively ends) with a trigger string, 
then $\alpha$ (respectively $\beta$) is the empty string.

Parsing $Q$ can be done in time linear in the length of $Q$. 

\begin{lemma}
We can represent $M$ with a data structure taking space (in words) proportional to the number of distinct phrases in $\D$.
Given a query pattern $Q$, this data structure returns NULL with high probability if $Q$ does not occur in $S$.
Otherwise ($Q$ occurs in $S$), it returns the partial encoding of $Q$.
In either case, this query takes $O(|Q|)$ time.
\end{lemma}
\begin{proof}
We keep the Karp-Rabin (KR) hashes of the phrases in $\D$, with the range of the KR hash function mapping to $[1..n^3]$ so the hashes each fit in $O(\log n)$ bits.  We also keep a constant-time map (implemented as a hash table with a hash function that's perfect for the phrases in $\D$) from the KR hashes of the phrases in $\D$ to their IDs, that returns NULL given any value that is not a KR hash of a phrase in $\D$.  We set $M$ to be the map composed with the KR hash function.

Given $Q$, we scan it to find the trigger strings in it, and convert it into a sequence of substrings consisting of: 
(a) the prefix $\alpha$ of $Q$ ending at the right end of the first trigger string in $Q$; 
(b) a sequence of PFP phrases, each starting and ending with a trigger string with no trigger string in between; and 
(c) the suffix $\beta$ of $Q$ starting at the left end of the last trigger string in $Q$.

We apply $M$ to every complete phrase in (b).  
If $M$ returns NULL for any complete phrase in (b), 
then $Q$ is not a substring of $S$, 
so we return NULL; 
otherwise, we return $\alpha$, the sequence of phrase IDs $M$ returned for the phrases in (b), and $\beta$.

Notice that, if a phrase in $Q$ is in $S$, then $M$ will map it to its lexicographic rank in $\D$; otherwise, the probability the KR hash of any particular phrase in $Q$ but not in $\D$ collides with the KR hash of a phrase in $\D$, is at most $n / n^3 = 1 / n^2$.  
It follows that, if $Q$ is a substring of $S$, then we return $Q$'s partial encoding; otherwise, we return NULL with high probability.
\end{proof}

\begin{corollary}
If we allow $O(|Q|)$ query time with high probability, then we can modify $M$ to always report NULL for $Q$ not occurring in $S$. 
\end{corollary}
\begin{proof}
We augment each Karp-Rabin (KR) hash stored in the hash table with the actual characters of its phrase such that we can check, character by character, whether a matched phrase of $Q$ is indeed in $\D$. In case of a collision we recompute the KR hashes of $\D$ and rebuild the hash table.
That is possible since we are free to choose different Karp-Rabin fingerprints for the phrases in $\D$.
\end{proof}

Continuing from our example above
where the trigger strings are {\tt AA}, {\tt CG} and {\tt TA},
suppose we have a given a query pattern $Q$,
\[Q [0..34] = \mathtt{CAGAAGAGTATCTCCTCGACATGTTGAAGACATAT}\,\]
we can compute the parse $Q$ to obtain the following
\[\mathtt{CAGAA}, \mathtt{AAGAGTA}, \mathtt{TATCTCCTCG}, \mathtt{CGACATGTTGAA}, \mathtt{AAGACATA}, \mathtt{TAT}.\]

Next, we use $M (h (\cdot))$ to map the complete phrases of this parse of $Q$ to their phrase IDs---which is their $\rank$ in $\D$. If any complete phrase maps to NULL then we know $Q$ does not occur in $T$. Using our example, we have the partial encoding
\[\mathtt{CAGAA}, 2, 4, 3, 1, \mathtt{TAT}.\]

Next, we consider all possible cases.  First, we consider the case that the last substring $\beta$ in our parse of $Q$ ends with a trigger string, which implies that it is a complete phrase.  Here, we can immediately start backward searching for the parse of $Q$ in the FM-index for $\P$.  Next, if $\beta$ is not a complete phrase then we backward search for $\beta$ in the FM-index for $S$.  If this backward search for $\beta$ returns nothing then we know $Q$ does not occur in $S$.  If the backward search for $\beta$ returns an interval in the \BWT{} of $\P$ that is not contained in the \BWT{} interval for a trigger string then $\beta$ does not start with a trigger string so $Q = \beta$ and we are done backward searching for $Q$.  

Finally, we consider the case when $\beta$ is a proper prefix of a phrase and the backward search for $\beta$ returns a $\BWT_S$ interval contained in the $\BWT_S$ interval for a trigger string. In our example, $\beta = \mathtt{TAT}$ and our backward search for $\beta$ in the FM-index for $S$ returns the interval $\BWT_S [31..32]$, which is the interval for the trigger string {\tt TA}.  Next,  we use $B$ to map the interval for $\beta$ in the $\BWT_S$ to the interval in the $\BWT_{\P}$ that corresponds to the cyclic shifts of $S$ starting with $\beta$.  

\begin{lemma}
We can store in space (in words) proportional to the number of distinct trigger strings in $S$ a data structure $B$ with which,
\begin{itemize}
\item given the lexicographic range of suffixes of $S$ starting with a string $\beta$ such that $\beta$ starts with a trigger string and contains no other trigger string, in $O(\log \log n)$ time we can find the lexicographic range of suffixes of $\P$ starting with phrases that start with $\beta$;
\item given a lexicographic range of suffixes of $\P$ such that the corresponding suffixes of $S$ all start with the same trigger string, in $O(\log \log n)$ time we can find the lexicographic range of those corresponding suffixes of $S$.
\end{itemize}
\end{lemma}
\begin{proof}
Let $B[0..n - 1]$ be a bitvector with 1s marking the lexicographic ranks of suffixes of $S$ starting with trigger strings.  There are at most as many runs of 1s in $B$ as there are distinct trigger strings in $S$, so we can store it in space proportional to that number and support rank and select operations on it in $O(\log \log n)$ time.

If $\BWT_S [i..j]$ contains the characters immediately preceding, in $S$, occurrences of a string $\beta$ that starts with a trigger string and contains no other trigger strings, then 
$\BWT_{\P} [B.\rank_1 (i)..B.\rank_1 (j)]$ contains the phrase IDs immediately preceding, in $\P$, the IDs of phrases starting with $\beta$.

If $\BWT_{\P} [i..j]$ contains the phrase IDs immediately preceding, in $\P$, suffixes of $\P$ such that the corresponding suffixes of $S$ all start with the same trigger string, then $\BWT_S [B.\select_1 (i + 1).. B.\select_1 (j + 1)]$ contains the characters immediately preceding the corresponding suffixes of $S$.

The correctness follows from Lemma~\ref{lemma:lexordercorrespondence}.
\end{proof}

Continuing with our example mapping $\BWT_S [31..32]$ yield the following interval:
\[\BWT_{\P} [B.\rank_1 (31), B.\rank_1 (32)] = \BWT_{\P} [4..5]\] as shown in Figure~\ref{fig:P_matrix}. Starting from this interval in $\BWT_{\P}$, we now backward search in the FM-index for $\P$ for the sequence of complete phrase IDs in the parse of $Q$.  In our example, we have the interval $\BWT_{\P} [4..5]$ which yields the following phrase IDs: $2\ 4\ 3\ 1$.

If this backward search in the FM-index for $\P$ returns nothing, then we know $Q$ does not occur in $S$.  Otherwise, it returns the interval in $\BWT_{\P}$ corresponding to cyclic shifts of $S$ starting with the suffix of $Q$ that starts with $Q$'s first complete phrase.  In our example, if we start with $\BWT_{\P} [4..5]$ and backward search for $2\ 4\ 3\ 1$ then we obtain $\BWT_{\P} [2]$, which corresponds to the cyclic shift
\[\mathtt{\tt AAGAGTATCTCCTCGACATGTTGAAGACATATGAT\$TCCAG}\]
of $S$ that starts with the suffix
\[\mathtt{AAGAGTATCTCCTCGACATGTTGAAGACATAT}\]
of $Q$ that is parsed into $2, 4, 3, 1, \mathtt{TAT}$.  

To finish our search for $Q$, we use $B$ to map the interval in $\BWT_{\P}$ to the corresponding interval
in the $\BWT_S$, which is the interval of rows in the \BWT{} matrix for $S$ which start with the suffix of $Q$ we have sought so far. In our example, we have that  $\BWT_{\P} [2]$ maps to  \[\BWT_S [B.\select_1 (2 + 1)] = \BWT_S [2]. \]
We note that our examples contain \BWT{} intervals with only one entry because our example is so small, but in general they are longer.  If the first substring $\alpha$ in our parse of $Q$ is a complete phrase then we are done backward searching for $Q$.  Otherwise, we start with this interval in $\BWT_S$ and backward search for $\alpha$ in the FM-index for $S$, except that we ignore the last $w$ last characters of $\alpha$ (which we have already sought, as they are also contained in the next phrase in the parse of $Q$). 

In our example, $\alpha = \mathtt{CAGAA}$ so, starting with $\BWT_S [2]$ we backward search for $\mathtt{CAG}$, which returns the interval $\BWT_S [14]$.  As shown in Figure~\ref{fig:T_matrix},
\[S [\SA [4]..n] = S [2..n] = \mathtt{CAGAAGAGTATCTCCTCGACATGTTGAAGACATATGAT\$}\]
does indeed start with
\[Q = \mathtt{CAGAAGAGTATCTCCTCGACATGTTGAAGACATAT}\,.\]  This concludes our explanation of {\tt count}.

To conclude, we give some intuition as to why we expect our  two-level FM-index to be faster in practice than standard backward search.  First, we note that standard backward search takes linear time in the length of $Q$ and usually uses at least one random access per character in $Q$.  Whereas, prefix-free parsing $Q$ takes linear time but does not use random access; backward search in the FM-index of $S$ is the same as standard backward search but we use it only for the first and last substrings in the parse of $Q$.  Backward search in the FM-index for $P$ is likely to use about $\lg |\D|$ random access for each complete phrase in the parse of $Q$: the \BWT{} of $\P$ is over an effective alphabet whose size is the number of phrases in $\D$.  Therefore, a balanced wavelet tree to support \rank{} on that \BWT{} should have depth about $\lg |\D|$ and we should use at most about one random access for each level in the tree.

In summary, if we can find settings of the prefix-free parsing parameters $w$ and $p$ such that
\begin{itemize}
\item most query patterns will span several phrases,
\item most phrases in those patterns are fairly long,
\item $\lg |\D|$ is significantly smaller than those phrases' average length,
\end{itemize}
then the extra cost of parsing $Q$ should be more than offset by using fewer random accesses. 
\section{Results} \label{sec:results}

We implemented our algorithm and measured its performance against all known competing methods.  We ran all experiments on a server with AMD EPYC 75F3 CPU with the Red Hat Enterprise Linux 7.7 (64bit, kernel 3.10.0). The compiler was g++ version 12.2.0. The running time and memory usage was recorded by SnakeMake benchmark facility~\cite{snakemake}.  We set a memory limitation of 128 GB of memory and a time limitation of 24 hours.

\subparagraph{Datasets.}
We used the following datasets. First, we considered sets of SARS-CoV-2 genomes taken from the NCBI website. We used three collections of $25,000$, $50,000$, and $100,000$ SARS-CoV-2 genomes from EMBL-EBI's COVID-19 data portal. Each collection is a superset of the previous. We denote these as {\tt SARS-25k}, and {\tt SARS-50k}, {\tt SARS-100k}. Next, we considered a single human reference genome, which we denote as {\tt GRCh38}, downloaded from NCBI. We report the size of the datasets as the number of characters in each in Table \ref{Tab:Construction-Usage}.  We denote $n$ as the number of characters.

\subparagraph{Implementation.}
We implemented our method in C++ 11 using the {\tt sdsl-lite} library~\cite{gbmp2014sea} and extended the prefix-free parsing method of Oliva, whose source code is publicly available here \url{https://github.com/marco-oliva/pfp}.  The source code for $\ours$ is available at \url{https://github.com/marco-oliva/afm}.

\begin{figure}[h]
     \centering
     \begin{subfigure}[b]{0.47\textwidth}
         \centering
         \includegraphics[width=\textwidth]{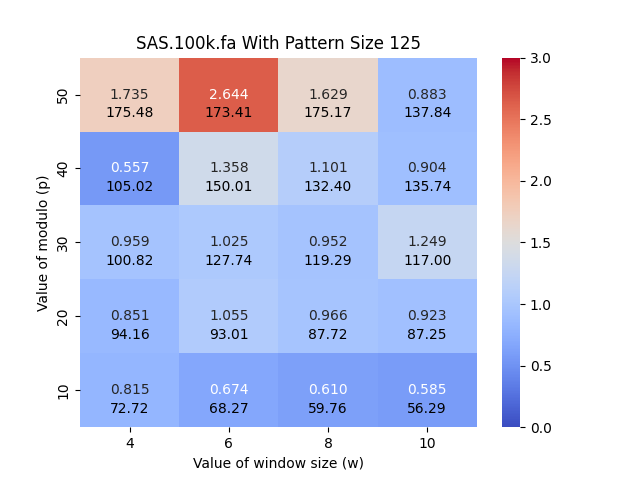}
     \end{subfigure}
     \begin{subfigure}[b]{0.47\textwidth}
         \centering
         \includegraphics[width=\textwidth]{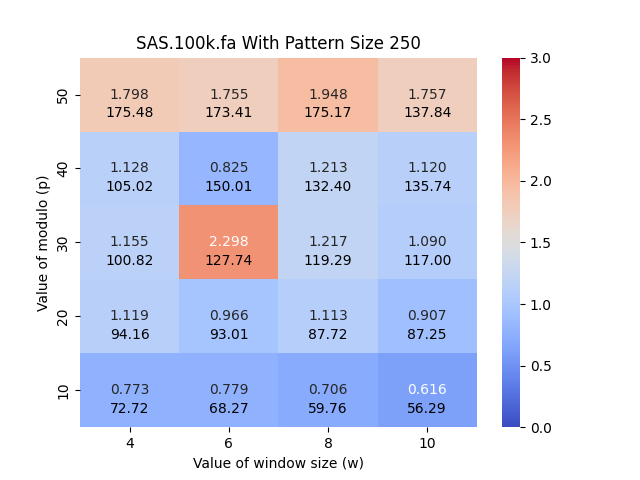}
     \end{subfigure}
     \\
     \begin{subfigure}[b]{0.47\textwidth}
         \centering
         \includegraphics[width=\textwidth]{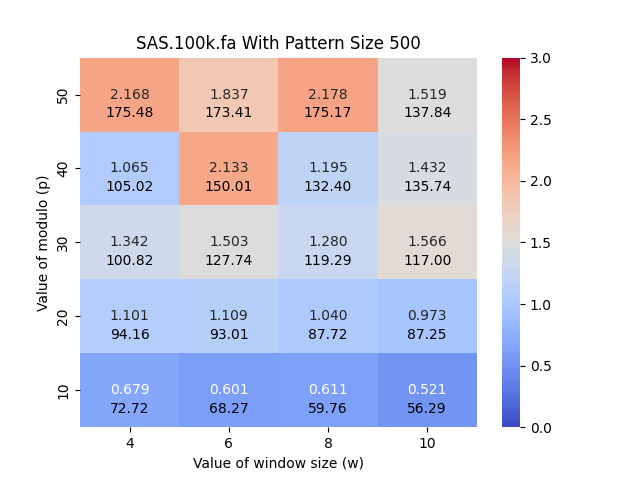}
     \end{subfigure}
     \begin{subfigure}[b]{0.47\textwidth}
         \centering
         \includegraphics[width=\textwidth]{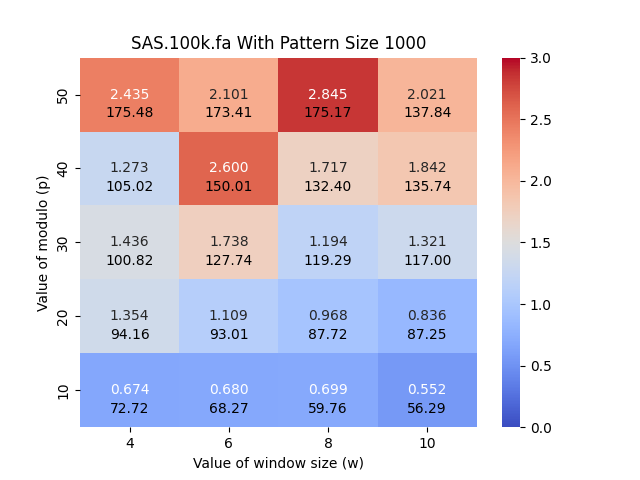}
     \end{subfigure}
        \caption{Illustration of the impact of $w$, $p$ and the length of the query pattern  on the acceleration of the FM-index. Here, we used {\tt SARS-100K} dataset and varied the length of the query pattern to be equal to 125, 250, 500, and 1000.  The y-axis corresponds to $p$ and the x-axis corresponds to $w$.  The heatmap illustrates the number of queries that can be performed in a CPU second with the acceleration versus the standard FM-index from {\tt sdsl}, i.e., $\ours$ / {\tt sdsl}. }
        \label{fig:accerlation}
\end{figure}

\subparagraph{Competing methods.}
We compared $\ours$ against the following methods the standard FM-index found in {\tt sdsl-lite} library~\cite{gbmp2014sea}, {\tt RLCSA}~\cite{siren2009compressed}, {\tt RLFM}~\cite{makinen2005succinct,makinen2004run}, and {\tt FIGISS} \cite{deng2022fm}.  We note that {\tt RLCSA} and {\tt FIGISS} have publicly-available source codes, while {\tt RLFM} is provided only as an executable. 
We performed the comparison by selecting 1,000 strings from the input file at random of the specified length, performing the {\tt count} operation on each query pattern, and measuring the time usage for all the methods under consideration. It is worth noting that {\tt FIGISS} and {\tt RLCSA} only support {\tt count} queries where the string is provided in an input text file. More specifically, the original {\tt FIGISS} implementation supports counting with the entire content of a file treated as a single pattern. To overcome this limitation, we modified the source code to enable the processing of multiple query patterns within a single file. In addition to the time consideration for {\tt count}, we measured the time and memory required to construct the data structure. 

\subsection{Acceleration versus Baseline}

In this subsection, we compare $\ours$ versus the standard FM-index in {\tt sdsl} with varying values of window size ($w$) and modulo value ($p$), and varying the length of the query pattern. We calculated the number of {\tt count} queries that were able to be performed in CPU second with $\ours$ versus the standard FM-index. We generated heatmaps that illustrate the number of {\tt count} queries of $\ours$ verses {\tt sdsl} for various lengths of query patterns, namely, 125, 250, 500, and 1,000. We performed this for each SARS-CoV-2 set of genomes. Figure \ref{fig:accerlation} shows the resulting heatmaps for {\tt SARS-100K}.  As shown in this figure, $\ours$ was between 2.178 and 2.845 times faster than the standard FM-index with the optimal values of $w$ and $p$.  In particular, an optimal performance gain of 2.6, 2.3, 2.2, and 2.9 was witnessed for pattern lengths of  125, 250, 500, and 1,000, respectively.  The $(w, p)$ pairs that correspond to these results are $(6,50)$, $(6,30)$, $(8,50)$, and $(8,50)$.  

Similar results were witnessed for {\tt SARS-25K} and  {\tt SARS-50K}.  For {\tt SARS-25K}, the optimal acceleration for pattern lengths of 125, 250, and 500 were 1.508, 2.261, 4.213, and 5.467, respectively (see Figure 6 in the Appendix).   For {\tt SARS-50K}, the optimal acceleration for pattern lengths of 125, 250, and 500 were 1.288, 2.984, 3.619, and 4.835, respectively (see Figure 7 in the Appendix).
For {\tt SARS-50K}, the optimal acceleration for pattern lengths of 125, 250, and 500 were 1.288, 2.984, 3.619, and 4.835, respectively (see Figure 7 in the Appendix). For {\tt GRCh38}, the optimal acceleration for pattern lengths of 125, 250, and 500 were 1.963, 1.805, 2.492, and 2.843, respectively (see Figure 8 in the Appendix).These results will guide our selection of $w$ and $p$ in the subsequent experiments.

\begin{figure}[htp]
    \includegraphics[width=.33\textwidth]{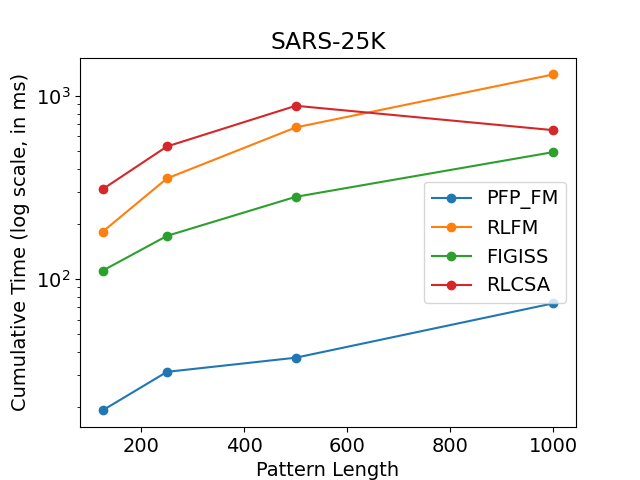}\hfill
    \includegraphics[width=.33\textwidth]{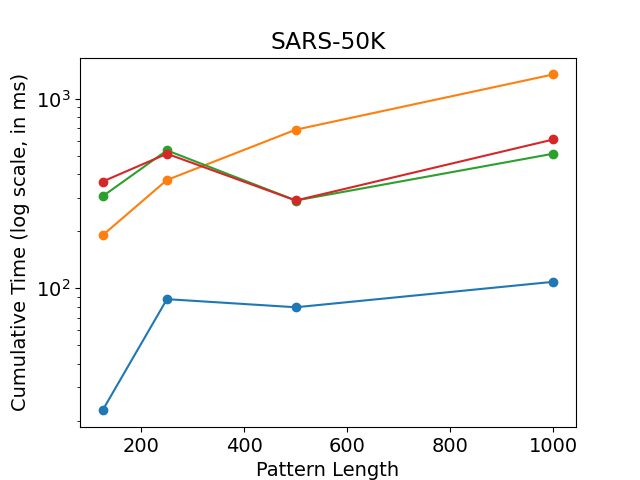}\hfill
    \includegraphics[width=.33\textwidth]{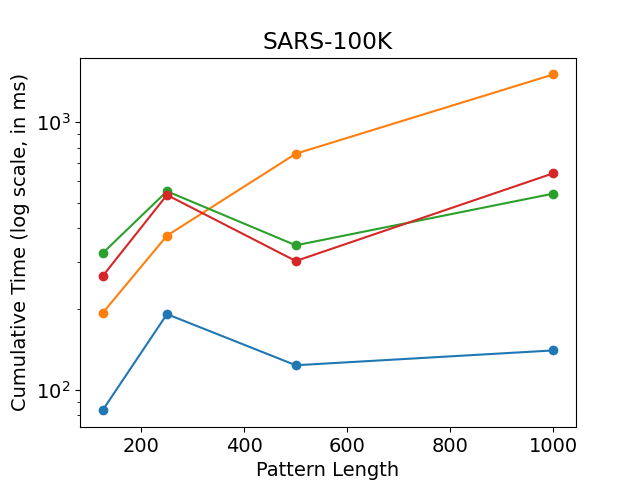}
     \caption{Illustration of the impact of the dataset size, and the length of the query pattern on the query time for answering {\tt count}. We vary the length of the query pattern to be equal to 125, 250, 500, and 1000, and report the times for {\tt SARS-25K}, {\tt SARS-50K}, and {\tt SARS-100K}.    We illustrate the cumulative time required to perform 1,000 {\tt count} queries. The y-axis is in $\log$ scale.  }
        \label{fig:count}
\end{figure}

\subsection{Results on SARS-CoV-2 Genomes}

We used the optimal parameters that were obtained from the previous experiment for this section. We constructed the index using these parameters for each SARS-CoV-2 dataset and assessed the time consumption for performing 1,000 {\tt count} queries using all competing methods and $\ours$.  We illustrate the result of this experiment in Figure \ref{fig:count}. It is clear from this $\ours$ consistently exhibits the lowest time consumption and a gradual, stable trend. For the {\tt SARS-25K} dataset, the time consumption of {\tt FIGISS} was between 451\% and 568\% higher than our method. And the time consumption of  {\tt RLCSA} and {\tt RLFM} was between 780\% and 1598\%, and 842\% and 1705\% more than $\ours$, respectively.  The performance of {\tt FIGISS} surpasses that of {\tt RLFM} and {\tt RLCSA} when using the SARS-25k dataset; however for the larger datasets {\tt FIGISS} and {\tt RLCSA} converge in their performance  .Neither method was substantially better than the other.  In addition, on the larger datasets, when the query pattern length was 125 and 250, {\tt RLFM} performed better than {\tt RLCSA} and {\tt FIGISS} but was slower for the other query lengths.   Hence, it is very clear that $\ours$ accelerates the performance of {\tt count} over all state-of-the-art methods. 

The gap in performance between $\ours$ and the competing methods increased with the dataset size.  For {\tt SARS-50K}, {\tt FIGISS}, {\tt RLCSA} and {\tt RLFM} were between 3.65 and 13.44, 3.65 and 16.08, and 4.25 and 12.39 times slower, respectively.
For {\tt SARS-100K}, {\tt FIGISS}, {\tt RLCSA} and {\tt RLFM} were between 2.81 and 3.86, 2.45 and 4.59, and 1.96 and 10.75 times slower, respectively.  

Next, we consider the time and memory required for construction; which is given in  Table \ref{Tab:Construction-Usage}. Our experiments revealed that all methods used less than 60 GB of memory on all SARS-CoV-2 datasets;  $\ours$ used the most memory with the peak being 54 GB on the {\tt SARS-100K} dataset. Yet, $\ours$ exhibited the most efficient construction time across all datasets for generating the FM-index, and this gap in the time grew with the size of the dataset.  More specifically, for the {\tt SARS-100K} dataset, $\ours$ used 71.04\%, 65.81\%, and 73.41\% less time compared to other methods. In summary, $\ours$ significantly accelerated the {\tt count} time, and had the fastest construction time. All methods used less than 60 GB, which is available on most commodity servers.

\subsection{Results on Human Reference Genome}
After measuring the time and memory usage required to construct the data structure across all methods using the {\tt GRCh38} dataset, we observed that $\ours$ exhibited has the second most efficient construction time but used the most construction space (71 GB vs. 26 GB to 45 GB). More specifically, $\ours$ was able to construct the index between 1.25 and 1.6 times faster than the {\tt FIGISS} and {\tt RLFM}.
\begin{figure}[h]
  \begin{minipage}[c]{0.6\textwidth}
     \includegraphics[width=\textwidth]{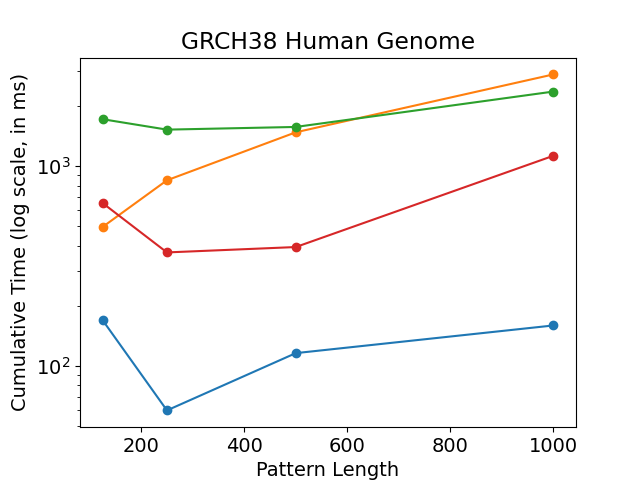}
\end{minipage}
  \begin{minipage}[c]{0.35\textwidth}

        \caption{Comparison of query times for \texttt{count} between the described solutions when varying the length of the query pattern. For each pattern length equal to 125, 250, 500, and 1000, we report the times for the {\tt GRCH38} dataset.    We plot the cumulative time required to perform 1,000 {\tt count} queries. The y-axis is in $\log$ scale. $\ours$ is shown in blue,  {\tt RLFM} is shown in orange, {\tt RLFM} is shown in red, and {\tt FIGISS} is shown in green.
        }
        \label{fig:human}
        \end{minipage}
\end{figure}

Next, we compare the performance of $\ours$ against other methods by performing 1,000 {\tt count} queries on, and illustrate the results in Figure \ref{fig:human}. Our findings demonstrate that $\ours$ consistently outperforms all other methods. Although {\tt RLCSA} shows better performance than {\tt RLFM} and {\tt FIGISS} when the pattern length is over 125 but is still 3.9, 6.2, 3.4, and 7.1 times slower than $\ours$. Meanwhile, the {\tt RLFM} method exhibits a steady increase in time usage, and it is 2.9, 14.2, 12.8, and 18.07 times slower than $\ours$. It is worth noting that the {\tt FIGISS} grammar is less efficient for non-repetitive datasets, as demonstrated in the research by Akagi et al.~\cite{Akagi2021GrammarIB}, which explains its (worse) performance on {\tt GRCh38} versus the  {\tt SARS-100K} dataset. Hence, {\tt FIGISS} is 10.1, 25.5, 13.6, and 14.8 times slower than $\ours$. These results are inline with the performance of our previous results, and demonstrate that $\ours$ has both competitive construction memory and time, and achieves a significant acceleration.

\begin{table}[h!]
\begin{center}
\begin{tabular}{ |c|c|c|c|c| } 
\hline
Dataset & $n$ & Method & Construction Memory (GB) & Construction Time (s) \\
\hline
\hline
\multirow{4}{6em}{{\tt SARS-25k}} & \multirow{4}{6em}{751,526,774} & {\tt RLCSA} &9.90  &322.85 \\ 
& &{\tt RLFM} &3.47  &363.74 \\ 
& &{\tt FIGISS} &4.89  &378.49 \\ 
& &{\tt \ours} &12.99  &117.29 \\
\hline
\multirow{4}{6em}{{\tt SARS-50k}} & \multirow{4}{6em}{1,503,252,577} &  {\tt RLCSA} &19.88  &679.89 \\ 
& & {\tt RLFM} &6.94  &701.36 \\ 
& &{\tt FIGISS} &12.44  &795.70 \\ 
& & {\tt \ours} &26.12  &233.04 \\
\hline
\multirow{4}{6em}{{\tt SARS-100k}} & \multirow{4}{6em}{3,004,588,730} &  {\tt RLCSA} &39.47  &1690.22 \\ 
& &{\tt RLFM} &25.01  &1432.16 \\ 
& &{\tt FIGISS} &25.57  &1840.80 \\ 
& & {\tt \ours} &53.90  &489.45 \\
\hline
\multirow{4}{6em}{{\tt GRCh38}} & \multirow{4}{6em}{3,189,750,467} &  {\tt RLCSA} &45.45  &924.60 \\ 
& & {\tt RLFM} &26.31  &1839.25 \\ 
& & {\tt FIGISS} &34.65  &1440.19 \\ 
& & {\tt \ours} &71.13  &1154.12 \\
\hline
\end{tabular}
\caption{Comparison of the construction performance with the construction time and memory for all datasets.  The number of characters in each dataset (denoted as $n$) is given in the second column.  The time is reported in seconds (s), and the memory is reported in gigabytes (GB). }
\label{Tab:Construction-Usage}
\end{center}
\end{table}

\section{Conclusion} 

In this work, we presented $\ours$ that shows significant acceleration over existing state-of-the-art methods.  Hence, this work begins to resolve a relatively long-standing issue in data structures as to how we can parse input that has no natural word boundaries in a manner that enables acceleration of the FM-index.  We note that it is possible to similarly augment \texttt{locate} queries since for that we need the suffix array samples only in the final step when matching $\alpha$ (or $\beta$ in case that $Q = \beta$), which can be done by the usually suffix array samplings for the FM-index. If $\alpha$ is empty, then we can instead match the first block of the pattern with the FM-index on $S$ and not on $\P$.  We leave this for future work.  With respect to practical applications, as reads are getting longer and more accurate, we will soon see an opportunity to apply accelerations of finding patterns that have length between 125 and 1,000. Hence, a larger area that warrants future consideration is accelerating the backward search with approaches such as $\ours$ for aligning Nanopore reads to a database.  Our last experiment shows significant acceleration with query patterns of length 1,000 to a full human reference genome, giving proof that the research community is in the position to begin such an endeavour.

\newpage
\appendix

\onecolumn
\section{Appendix}

\begin{figure}[h]
     \centering
     \begin{subfigure}[b]{0.47\textwidth}
         \centering
         \includegraphics[width=\textwidth]{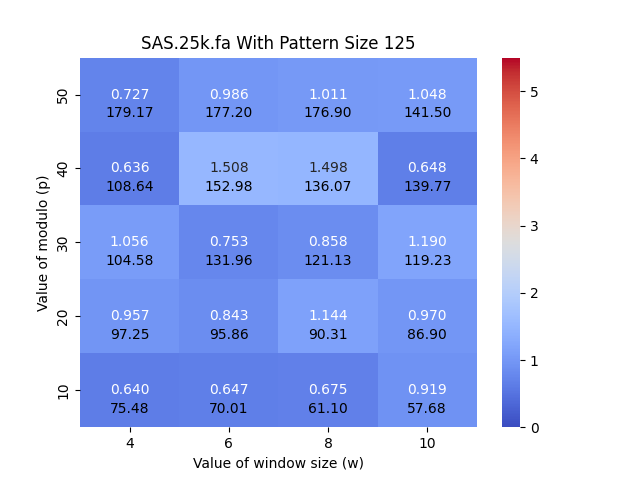}
     \end{subfigure}
     \begin{subfigure}[b]{0.47\textwidth}
         \centering
         \includegraphics[width=\textwidth]{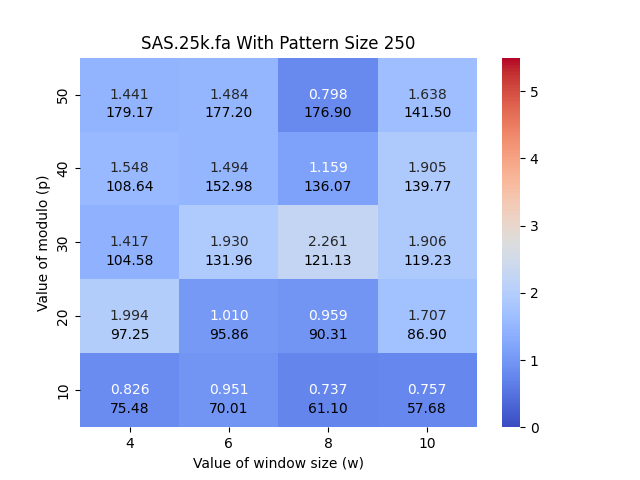}
     \end{subfigure}
     \\
     \begin{subfigure}[b]{0.47\textwidth}
         \centering
         \includegraphics[width=\textwidth]{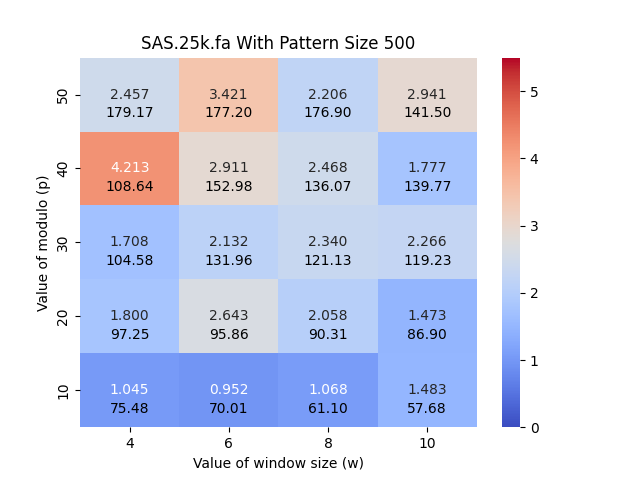}
     \end{subfigure}
     \begin{subfigure}[b]{0.47\textwidth}
         \centering
         \includegraphics[width=\textwidth]{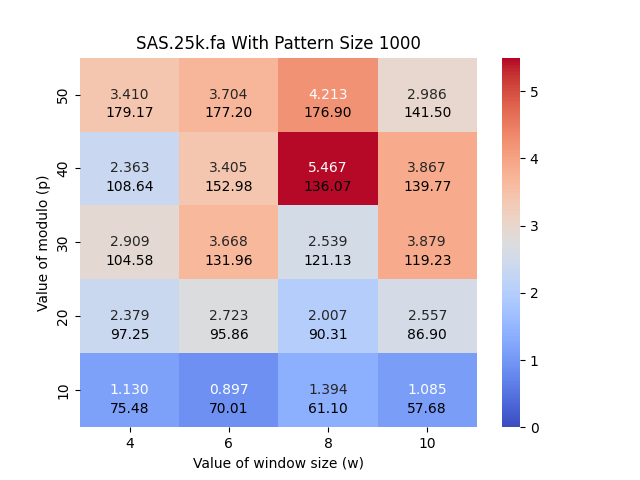}
     \end{subfigure}
        \caption{Illustration of the impact of $w$, $p$ and the length of the query pattern  on the acceleration of the FM-index. Here, we used {\tt SARS-25K} dataset and varied the length of the query pattern to be equal to 125, 250, 500, and 1000.  The y-axis corresponds to $p$ and the x-axis corresponds to $w$.  The heatmap illustrates the number of queries that can be performed in a CPU second with the acceleration verses the standard FM-index from {\tt sdsl}, i.e., $\ours$ / {\tt sdsl}. }
        
\end{figure}

\begin{figure}[h]
     \centering
     \begin{subfigure}[b]{0.47\textwidth}
         \centering
         \includegraphics[width=\textwidth]{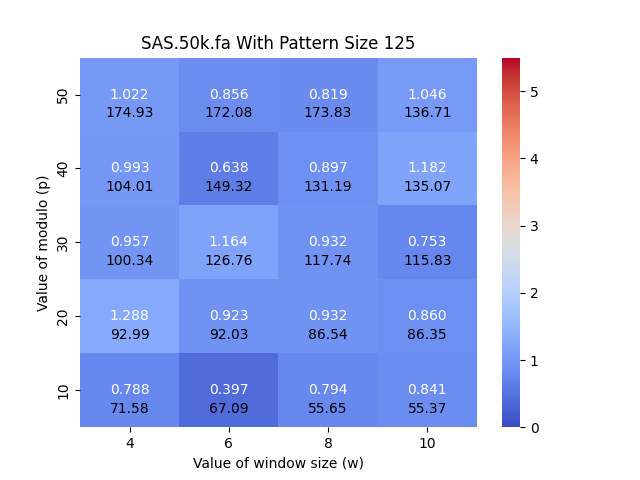}
     \end{subfigure}
     \begin{subfigure}[b]{0.47\textwidth}
         \centering
         \includegraphics[width=\textwidth]{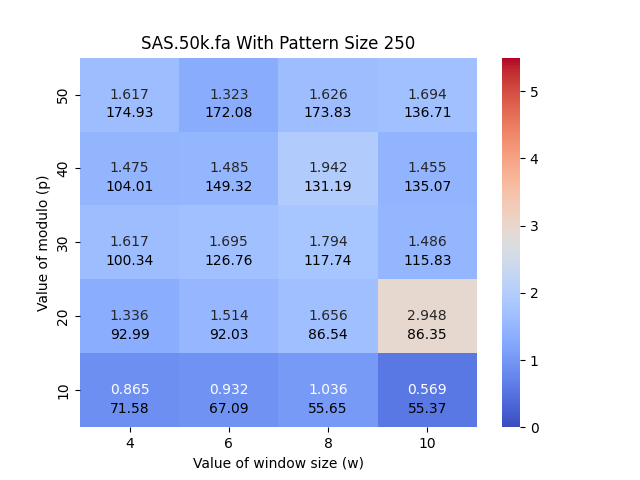}
     \end{subfigure}
     \\
     \begin{subfigure}[b]{0.47\textwidth}
         \centering
         \includegraphics[width=\textwidth]{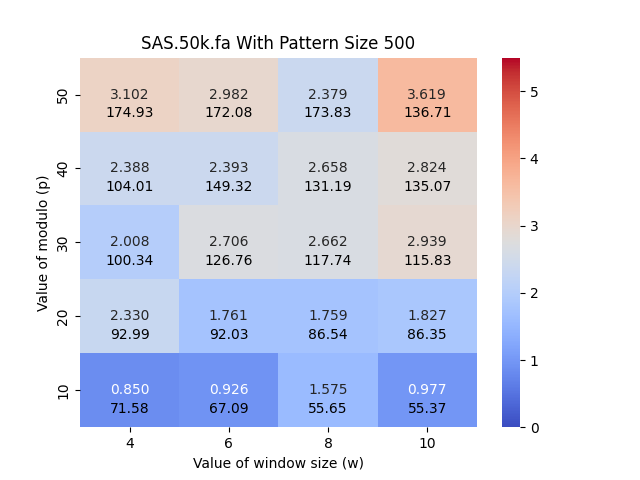}
     \end{subfigure}
     \begin{subfigure}[b]{0.47\textwidth}
         \centering
         \includegraphics[width=\textwidth]{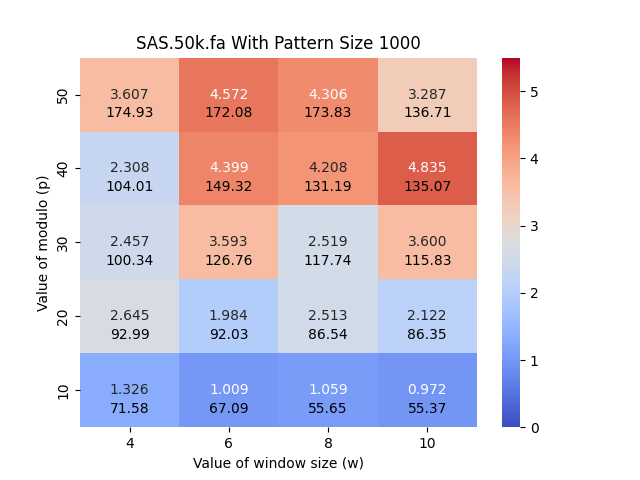}
     \end{subfigure}
        \caption{Illustration of the impact of $w$, $p$ and the length of the query pattern  on the acceleration of the FM-index. Here, we used {\tt SARS-50K} dataset and varied the length of the query pattern to be equal to 125, 250, 500, and 1000.  The y-axis corresponds to $p$ and the x-axis corresponds to $w$.  The heatmap illustrates the number of queries that can be performed in a CPU second with the acceleration verses the standard FM-index from {\tt sdsl}, i.e., $\ours$ / {\tt sdsl}. }
\end{figure}

\begin{figure}[h]
     \centering
     \begin{subfigure}[b]{0.47\textwidth}
         \centering
         \includegraphics[width=\textwidth]{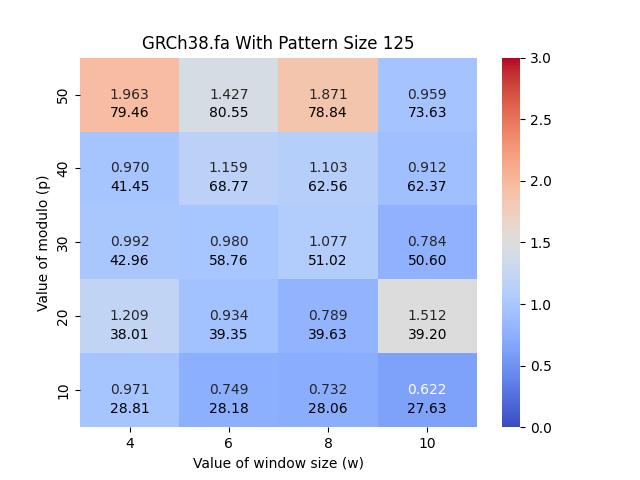}
     \end{subfigure}
     \begin{subfigure}[b]{0.47\textwidth}
         \centering
         \includegraphics[width=\textwidth]{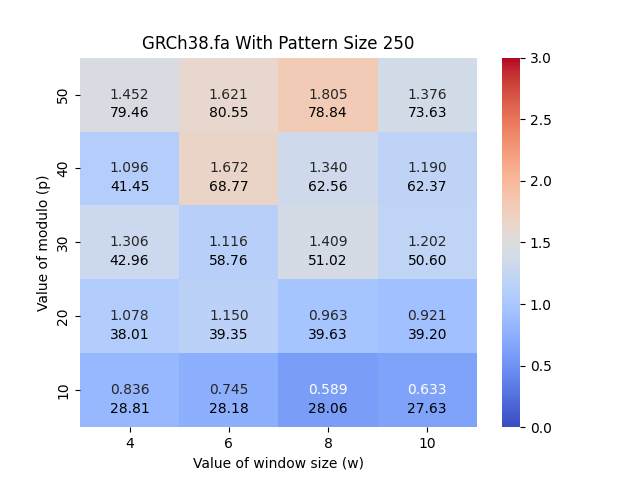}
     \end{subfigure}
     \\
     \begin{subfigure}[b]{0.47\textwidth}
         \centering
         \includegraphics[width=\textwidth]{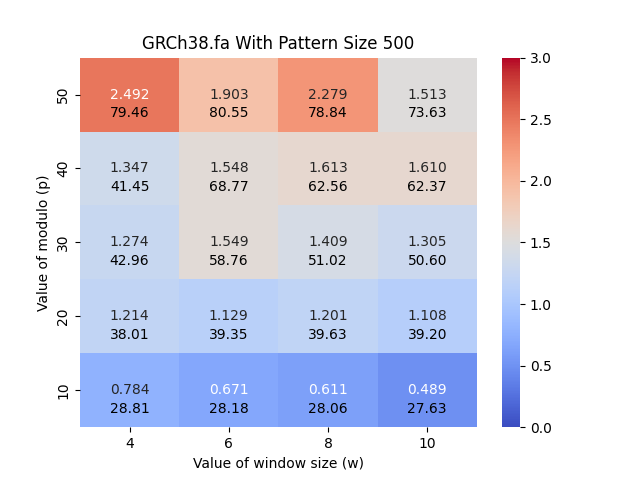}
     \end{subfigure}
     \begin{subfigure}[b]{0.47\textwidth}
         \centering
         \includegraphics[width=\textwidth]{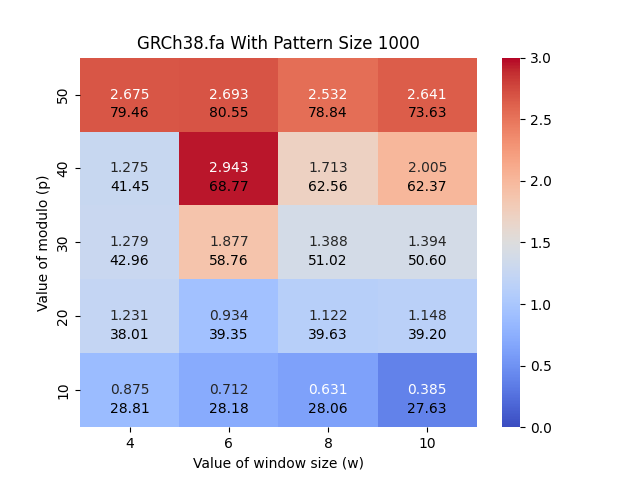}
     \end{subfigure}
        \caption{Illustration of the impact of $w$, $p$ and the length of the query pattern  on the acceleration of the FM-index. Here, we used {\tt GRCh38} dataset and varied the length of the query pattern to be equal to 125, 250, 500, and 1000.  The y-axis corresponds to $p$ and the x-axis corresponds to $w$.  The heatmap illustrates the number of queries that can be performed in a CPU second with the acceleration verses the standard FM-index from {\tt sdsl}, i.e., $\ours$ / {\tt sdsl}. }
\end{figure}
\end{document}